\documentclass[a4paper,11pt]{scrartcl}

\let\phi\varphi
\let\epsilon\varepsilon
\let\rho\varrho
\let\tilde\widetilde

\usepackage{color}
\usepackage[utf8]{inputenc}
\usepackage[english]{babel}
\usepackage[T1]{fontenc}
\usepackage{ifpdf}
\ifpdf
  \usepackage[pdftex]{graphicx}
  \graphicspath{{./Fig_pdf/}{./Pdf/}}
  \DeclareGraphicsExtensions{.pdf}
\else
  \usepackage[dvips]{graphicx}
  \graphicspath{{./Fig_eps/}{./Eps/}}
  \DeclareGraphicsExtensions{.ps,.eps}
\fi

\usepackage{amsmath}
\usepackage{amssymb}
\usepackage{amsfonts}
\usepackage{latexsym}
\usepackage{xspace}
\usepackage{scrtime}
\usepackage{mdwlist}
\usepackage{enumerate}

\usepackage[amsmath,thmmarks,noconfig]{ntheorem}
\theoremseparator{\quad}
\newtheorem{lemma}{Lemma}
\newtheorem{theorem}[lemma]{Theorem}

\newtheorem{corollary}[lemma]{Corollary}
\theorembodyfont{\upshape}

\qedsymbol{\quad\ensuremath{\Box}}
\theoremstyle{nonumberplain}
\theoremheaderfont{\normalfont\itshape}
\theorembodyfont{\upshape}
\theoremsymbol{\quad\ensuremath{\Box}}
\theoremseparator{.}
\newtheorem{proof}{Proof}

\ifpdf
  \usepackage[pdftex]{hyperref}
\else
  \usepackage[dvips]{hyperref}
\fi

\usepackage{listings}
\title{An Optimal Algorithm for the Indirect Covering Subtree Problem}
\author{Joachim Spoerhase}

\newcommand{\RR}{\mathbb{R}}
\newcommand{\QQ}{\mathbb{Q}}
\newcommand{\NN}{\mathbb{N}}

\begin{document}

\maketitle

\begin{abstract}
  We consider the indirect covering subtree problem (Kim et al.,
  1996).  The input is an edge weighted tree graph along with
  customers located at the nodes.  Each customer is associated with a
  radius and a penalty.  The goal is to locate a tree-shaped facility
  such that the sum of setup and penalty cost is minimized.  The setup
  cost equals the sum of edge lengths taken by the facility and the
  penalty cost is the sum of penalties of all customers whose distance
  to the facility exceeds their radius.  The indirect covering subtree
  problem generalizes the single maximum coverage location problem on
  trees where the facility is a node rather than a subtree.  Indirect
  covering subtree can be solved in $O(n\log^2 n)$ time (Kim et al.,
  1996).  A slightly faster algorithm for single maximum coverage
  location with a running time of $O(n\log^2n/\log\log n)$ has been
  provided (Spoerhase and Wirth, 2009).  We achieve time $O(n\log n)$
  for indirect covering subtree thereby providing the fastest known
  algorithm for both problems.  Our result implies also faster
  algorithms for competitive location problems such as
  $(1,X)$-medianoid and $(1,p)$-centroid on trees.  We complement our
  result by a lower bound of $\Omega(n\log n)$ for single maximum
  coverage location and $(1,X)$-medianoid on a real-number RAM model
  showing that our algorithm is optimal in running time.

\end{abstract}

\section{Problem Definitions and Related Work}

We are given a tree $T=(V,E)$ along with edge costs $c\colon
E\rightarrow\RR_{\geq 0}$ inducing a distance function $d\colon
V\times V\rightarrow\RR_{\geq 0}$.  With each node $u$ we associate a
non-negative \emph{penalty} $\pi(u)$.  Let $Y$ be a subtree of $T$.
Then $c(Y)$ denotes the \emph{setup cost} of $Y$ and is given by the
sum $\sum_{e\in E(Y)}c(e)$ of its edge costs.  A node is \emph{covered
  directly by $Y$} if it lies in $Y$.  If some node $u$ is not covered
directly it imposes the penalty $\pi(u)$ on $Y$.  The \emph{direct
  covering subtree problem} \cite{kim-etal:location-tree-shaped} asks
for a subtree $Y$ such that the sum of setup cost $c(Y)$ and the total
penalty $\sum_{u\notin V(Y)}\pi(u)$ is minimized.

The \emph{indirect covering subtree problem} goes one step further.  A
node $u$ is said to be covered (indirectly) if it lies within a given
distance from $Y$.  Again, a penalty is imposed on $Y$ if it does not
cover $u$. More formally, we assign to each node $u$ some
\emph{radius} $\rho(u)$.  The \emph{penalty imposed on $Y$ by $u$} is
given as
\begin{displaymath}
  p(u,Y):=
  \begin{cases}
    0 & \text{ if } d(u,Y)\leq \rho(u)\\
    \pi(u) & \text{ otherwise\,.}
  \end{cases}
\end{displaymath}
If $U\subseteq V$ is a set of nodes then $p(U,Y):=\sum_{u\in U}
p(u,Y)$ is the penalty imposed on~$Y$ by~$U$.  The total penalty
imposed on $Y$ is given by $p(Y):=p(V,Y)$.  The indirect covering
subtree problem \cite{kim-etal:location-tree-shaped} asks for a
subtree $Y$ of $T$ such that the total cost $c(Y)+p(Y)$, given by the
sum of setup and penalty cost, is minimum among all subtrees of $T$.

If we require that $Y$ be a node rather than a subtree we obtain the
\emph{single maximum coverage location problem}
\cite{megiddo:maximum-coverage-location,spoerhase+wirth:single-max-coverage}.
It is not hard to see that single maximum coverage location is a
special case of indirect covering subtree.  (Scale all edge lengths
and radii with a sufficiently large factor while leaving the penalties
unchanged.)

\subsection{Related Work and Previous Results}

The \emph{multiple} maximum coverage location problem allows the
placement of a an arbitrary set of $r$ nodes.  On general graphs this
problem is NP-hard \cite{megiddo:maximum-coverage-location} while it
can be solved in time $O(rn^2)$ on trees \cite{Tamir:p-median-tree}.
This leads to an $O(n^2)$ algorithm for the \emph{single} maximum
coverage location problem on trees by setting $r=1$.  Kim et al.\
\cite{kim-etal:location-tree-shaped} provide a faster algorithm
running in $O(n\log^2n)$.  Their algorithm works even for the more
general indirect covering subtree problem.  Recently a slightly faster
algorithm for single maximum coverage location with time
$O(n\log^2n/\log\log n)$ has been reported
\cite{spoerhase+wirth:single-max-coverage}. Finally, we remark that
\emph{direct} covering subtree can be solved in linear time
\cite{kim-etal:location-tree-shaped}.

\subsection{Contribution and Outline of this Paper}

In this paper we show that indirect covering subtree can be solved in
$O(n\log n)$.  This improves upon the previously best algorithms for
this problem and single maximum coverage location on trees.  Our
result also implies faster algorithms for the $(1,X)$-medianoid
problem and the $(1,p)$-centroid problem on trees.  Specifically, we
obtain an $O(n\log n)$ algorithm for $(1,X)$-medianoid and $O(n^2\log
n\log w(T))$ and $O(n^2\log n\log w(T)\log D)$ algorithms for the
discrete and absolute $(1,p)$-centroid problems on trees,
respectively.  Here, $w(T)$ denotes the total weight of the tree and
$D$ is the maximum edge length.  The previously best algorithms are
slower by factor of $O(\log n/\log\log n)$
\cite{spoerhase+wirth:rp-centroid}.

Our algorithm employs the same dynamic programming framework used by
Kim et al.\ \cite{kim-etal:location-tree-shaped}.  However, we improve
one of their core routines by using a more sophisticated technique to
subdivide trees.  This technique, called two-terminal subtree
subdivision (TTST), is a simplification of the recursive coarsening
strategy \cite{spoerhase+wirth:single-max-coverage} used for solving
single maximum coverage location on a tree.  The key source of our
improvement is that we manage to avoid explicitly sorting the nodes
according to their distances and radii during the recursion, which has
been necessary in the coarsening approach and also in the original
algorithm of Kim et al.  One further advantage of our algorithm is
that it is a lot simpler than the recursive coarsening algorithm.

The two-terminal subtree technique has proved successful also for
other location problems
\cite{spoerhase+wirth:relaxed-mgf-dam,spoerhase+wirth:jda-stackelberg}.
I believe that there are further problem classes where it can be
applied.

The paper is organized as follows.  In
Section~\ref{sec:algorithm-kim-et} we briefly outline the algorithm of
Kim et al.  This is necessary, since our result relies on an
improvement of a subroutine of that algorithm.  The improved routine
is then described in Section~\ref{sec:an-ovlog-v}.  In Section
\ref{sec:matching-lower-bound} we provide a matching lower bound on
the running time needed to solve indirect covering subtree.  Finally,
we discuss implications on related problems such as competitive
location problems in Section~\ref{sec:impl-relat-probl}.

\section{The Algorithm of Kim et al.}\label{sec:algorithm-kim-et}

In the sequel we will briefly describe the algorithmic approach of Kim
et al.\ \cite{kim-etal:location-tree-shaped} for solving the indirect
covering subtree problem.

Let's first fix some conventions and notations.  We assume that the
input tree $T$ is rooted at some distinguished node $s$.  For
technical reasons we shall adopt the convention that $s$ is the father
of itself.  Let $v$ be an arbitrary node.  Then $f(v)$ denotes the
father of $v$.  We write $T_v$ for the subtree of $T$ hanging from $v$
and $T_v^+$ for the union of $T_v$ with the edge $(v,f(v))$.

Kim et al.\ reduce the solution of the problem to the computation of
the values $p(v), p(T_v,v)$ and $p(T_v,f(v))$ for all nodes $v$.  They
show that one can determine an optimum to the subtree location problem
in linear time once these values have been precomputed for all nodes
$v$.

To convince ourselves, assume that we have computed the values $p(v)$,
$p(T_v,v)$ and $p(T_v,f(v))$ for all $v\in V$.  Then define
\begin{equation*}
  C(v):=\min\{\,c(Y)+p(T_v,Y)\mid Y\text{ is subtree of } T_v\text{ containing
  } v\,\}\, ,
\end{equation*}
and
\begin{equation*}
  C^+(v):=\min\{\,c(Y)+p(T_v,Y)\mid Y\textnormal{ is subtree of }
  T_v^+ \text{ containing } f(v)\,\}\,.
\end{equation*}

It is not hard to see that the optimum cost can now be expressed by
\begin{equation*}
  \min_{v\in V} (C(v)+p(v)-p(T_v,v)\,).
\end{equation*}
Moreover, the $C(\cdot)$- and $C^+(\cdot)$-values can be computed in
linear time by means of a simple bottom-up dynamic programming
approach.  To this end assume that $v$ is a leaf of $T$ then
\begin{equation*}
C(v)=0 \text{\quad and\quad } C^+(v)=\min\{p(v,f(v)), c(v,f(v))\}\, .
\end{equation*}
Otherwise, we have
\begin{equation*}
  C(v)=\sum_{u\text{ is son of } v} C^+(u)\, ,
\end{equation*}
and
\begin{equation*}
    C^+(v)=\min\{C(v)+c(v,f(v)), p(T_v,f(v))\}\, .
\end{equation*}

From this equations it follows that an optimal solution can be
determined in linear time in a bottom-up fashion once the values
$p(v), p(T_v,v)$ and $p(T_v,f(v))$ have been computed for all $v\in
V$.  Kim et al.\ suggest an algorithm with running time $O(n\log^2 n)$
to compute these values. 

This algorithm is based on the so-called \emph{bitree model}.  In this
model, each (undirected) edge $(u,v)$ of the input tree is replaced
with two anti-parallel, directed arcs $(u,v)$, $(v,u)$.  We call the
resulting tree $T'$ \emph{bitree} of $T$.  With each arc $(u,v)$ of
the bitree we associate a cost $c_{T'}(u,v)$ representing the length
of this arc.  But in contrast to the edges of the input tree $T$ we
allow these costs to be negative and asymmetric.  This induces a
distance function $d_{T'}\colon V\times V\rightarrow\QQ$ where
$d_{T'}(u,v)$ is the length of the unique $u$-$v$-path in $T'$.  Now
we define the penalty cost $p'(u,v)$ imposed on $v$ by $u$ to be zero
if $d_{T'}(u,v)\leq\rho(u)$ and $\pi(u)$ otherwise.  We set
$p'(v)=\sum_{u\in V}p'(u,v)$.

The algorithm of Kim et al.\ is based on a subroutine for efficiently
computing $p'(v)$ for all nodes $v$ on a given bitree $T'$.  By means
of such a subroutine it is then possible to calculate the values
$p(v)$, $p(T_v,v)$ and $p(T_v,f(v))$ for all $v$ in the input tree
$T$.  It follows from the above discussion that the knowledge of these
values enables us to identify an optimal tree-shaped facility.

It remains to explain how we can employ such a subroutine to determine
$p(v)$, $p(T_v,v)$ and $p(T_v,f(v))$ for all nodes $v$ of the input
tree $T$ which, in turn, is sufficient to build an optimal tree-shaped
facility.

First we describe how we can determine $p(\cdot)$.  For this purpose
we simply set $c_{T'}(u,v):=c_{T'}(v,u):=c_T(u,v)$ for all edges
$(u,v)$ of the input tree $T$.  It is then immediately clear that
$p(v)=p'(v)$ for all $v\in V$.

In order to compute $p(T_v,v)$ for all $v\in V$ we set
$c_{T'}(u,f(u)):=c_T(u,f(u))$ and $c_{T'}(f(u),u):=-\infty$ for all
$u\neq s$.  This construction ensures that the penalty cost $p'(u,v)$
is always zero if $u$ is \emph{not} a descendant of $v$.  Thus
$p(T_v,v)=p'(v)$ holds for this construction.

Finally, we wish to determine $p(T_v,f(v))$.  To this end we introduce
on each edge $(v,f(v))$ of $T$ a new node $f'(v)$ such that edge
$(v,f'(v))$ has length $c_T(v,f(v))$ and edge $(f'(v),f(v))$ has
length zero.  This increases the number of nodes to $2n-1$.  We set
$\pi(f'(v))$ and $\rho(f'(v))$ to zero.  It is easy to see that
$p(T_v,f(v))$ in the original tree equals $p(T_u,u)$ in the newly
constructed tree where $u:=f'(v)$.  Hence the problem of computing
$p(T_v,f(v))$ for all nodes $v$ can be reduced to the problem of
computing $p(T_v,v)$, which has been described before.

Kim et al.\ provide a subroutine to compute the $p'(\cdot)$-values on
a bitree with $n'$ nodes in $O(n'\log ^2n')$ time which yields
immediately.

\begin{theorem}[\cite{kim-etal:location-tree-shaped}]
  The indirect covering subtree problem on a tree can be solved in
  $O(n\log^2n)$.\qed{}
\end{theorem}

\section{An \emph{O}(\emph{n}$\;$log \emph{n}) Algorithm}\label{sec:an-ovlog-v}

In this section we describe an algorithm for the indirect covering
subtree problem with running time $O(n\log n)$.

Our algorithm uses the algorithmic framework of Kim et al.\ described
in Section~\ref{sec:algorithm-kim-et}.  Specifically, we will provide
an improved routine for computing the values $p'(\cdot)$ on a given
bitree in $O(n\log n)$ which can then be extended to an algorithm with
the same asymptotic running time for solving indirect covering
subtree.

The basic approach of the routine of Kim et al.\ for computing
$p'(\cdot)$ is divide-and-conquer.  It partitions the node set $V$
into two sets $V_1, V_2$ of bounded size such that both induce
subtrees and have exactly one node (called centroid
\cite{kang-ault:properties-centroid-tree}) in common.  Then it sorts
the sets $V_i$ and computes, by means of a clever merge-and-scan
procedure, for all $v\in V_i$ the penalties $p'(v,V_j)$ of the users
in $V_j$ where $j\neq i$.  Applying the routine recursively to the
sub-bitrees induced by $V_1,V_2$ one can determine the
$p'(v,V_i)$-values also for each $v\in V_i$.  Finally, one obtains the
total penalty $p'(v,V)$ of any node $v\in V$ by adding $p'(v,V_1)$ and
$p'(v,V_2)$.

Our routine proceeds in a similar way but uses a more sophisticated
subdivision, which allows us to avoid the explicit sorting thereby
supressing the additional $\log$-factor.  Spoerhase and Wirth
\cite{spoerhase+wirth:relaxed-mgf-dam,spoerhase+wirth:jda-stackelberg}
used an analogous subdivision technique for solving competitive
location problems on undirected trees.

Consider the (undirected) input tree $T=(V,E)$. We may assume that $T$
has maximum degree three.  Otherwise, we can split nodes of larger
degree by introducing suitable zero-length edges and zero-weighted
nodes.  Let $T'$ be the bitree corresponding to $T$.

If $s$ and $t$ are distinct nodes then $T'_{st}$ denotes the maximal
sub-bitree of $T'$ having $s$ and $t$ as leaves. Let $V_{st}$ be the
node set of $T'_{st}$.  We call $s$ and $t$ \emph{terminals} and
$T'_{st}$ \emph{two-terminal sub-bitree} (TTSB).

Our algorithm divides the input bitree recursively into TTSBs.  Since
we are dealing with a degree-bounded bitree we can subdivide any TTSB
$S$ into at most five TTSBs, called \emph{child TTSBs}.  Each of these
child TTSBs has at most $\frac12|S|+1$ nodes.

\begin{lemma}\label{lem:subdiv-degree-ttsb}
  Let $S$ be a TTSB with maximum degree three.  Then $S$ can be
  partitioned into at most five edge-disjoint TTSBs each of which
  having at most $\frac12|S|+1$ nodes.  This subdivision can be
  computed in $O(|S|)$ time.
\end{lemma}
\begin{proof}
  Let $S$ be a TTSB with maximum degree three and terminals $u$ and
  $v$.  Let $m$ be the unweighted median of $S$, which can be computed
  in $O(|S|)$ by means of Goldman's algorithm
  \cite{goldman:optimal-center-location}.  (All node and arc weights
  are temporarily set to one throughout this proof.)  It is a
  well-known fact that $m$ has the following property: Each of the
  connected components of $S-m$ has at most $\frac12|S|$ nodes.
  Hence, if $m$ lies on path $P(u,v)$ then $S-m$ contains at most
  three components that form the desired subdivision (confer left part
  of Figure~\ref{fig:tree-subdivision}).  If $m$ does not lie on
  $P(u,v)$ then consider the node $m'$ on $P(u,v)$ that is closest to
  $m$ (confer right part of Figure~\ref{fig:tree-subdivision}).  Then
  $S-\{m,m'\}$ has at most five connected components.  All of the
  child TTSBs obtained this way have clearly at most $\frac12|S|+1$
  nodes.
\end{proof}
\begin{figure}[htp]
  \centering
  \input{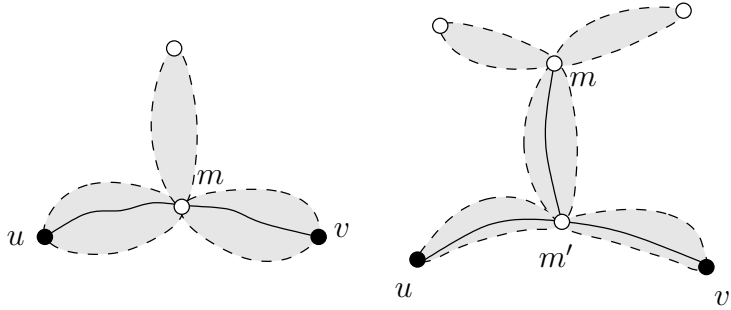}
  \caption{The two cases in the subdivision of a TTSB.}
  \label{fig:tree-subdivision}
\end{figure}

Consider a TTSB $T'_{st}$.  We introduce the lists $L_{d,s}(T'_{st})$
and $L_{\rho,s}(T'_{st})$.  Both lists contain all nodes $v$ of
$T'_{st}$ sorted in increasing order with respect to the values
$d_{T'}(s,v)$ and $\rho(v)-d_{T'}(v,s)$, respectively.  The lists
$L_{d,t}(T'_{st})$ and $L_{\rho,t}(T'_{st})$ are defined symmetrically.

The algorithm computes $p'(v,T'_{st})$ for all $v\in T'_{st}$ as well
as the four lists $L_{d,s}(T'_{st})$, $L_{d,t}(T'_{st})$,
$L_{\rho,s}(T'_{st})$ and $ L_{\rho,t}(T'_{st})$ for any TTSB
$T'_{st}$ occurring during the recursion.  We shall see that these
information can be propagated inductively from child towards parent
TTSBs such that we will have computed $p'(\cdot, T')=p'(\cdot)$ at the
top of the recursion.

To this end consider an arbitrary TTSB $S=T'_{st}$ being subdivided
into at most five child TTSBs $S_i$ with terminals $s_i$, $t_i$.
Moreover assume that we have already computed $p'(\cdot)$ and the four
lists corresponding to $S_i$ for all $S_i$.

We start with computing $L_{d,s}(S)$.  To this end we maintain a list
$L$ which is initialized with an empty list.  Now we perform the
following operations for all child TTSBs $S_i$: Assume that $s_i$ is
the terminal of $S_i$ closest to $s$.  Then the list $L_{d,s_i}(S_i)$
contains all nodes $v\in S_i$ with associated sorting keys
$d_{T'}(s_i,v)$.  Now we create a copy $L'$ of this list and add the
value $d_{T'}(s,s_i)$ to all sorting keys which does not affect its
order.  As a result $L'$ contains all nodes $v$ of $S_i$ sorted with
respect to their distance $d_{T'}(s,v)$ from terminal $s$.  Finally we
merge $L$ with $L'$.  After having carried this out for all child
TTSBs $S_i$ the list $L$ equals the list $L_{d,s}(S)$ we are looking
for.  The list $L_{\rho,s}(S)$ is computed very similarly with the
difference that we \emph{subtract} the value $d_{T'}(s_i,s)$ from the
sorting keys when handling the list $L_{\rho,s}(S_i)$.  The respective
lists for terminal $t$ are computed symmetrically.  The total running
time for computing the four lists associated with $S$ is $O(|S|)$
since we handle a constant number of child TTSBs.

We are now going to explain how $p'(v,S)$ can be determined for all
$v\in S$.  To this end assume that $v$ is contained in some $S_i$.
Since we already know $p'(v,S_i)$ by the inductive hypothesis it
suffices to determine $p'(v,S_j)$ for all $S_j\neq S_i$ and to add
these values to $p'(v,S_i)$.  Consider an arbitrary $S_j\neq S_i$ and
assume that $s_i,s_j$ are the terminals of these TTSBs closest to each
other.  We create a copy $L'$ of list $L_{\rho,s_j}(S_j)$ and
substract the distance $d_{T'}(s_j,s_i)$ from all sorting keys in this
list.  As a result $L'$ contains all nodes $u$ of $S_j$ sorted with
respect to the key $\rho(u)-d_{T'}(u,s_i)$.  At this point we can
compute $p'(v,S_j)$ for \emph{all} $v\in S_i$ by using the
merge-and-scan procedure of Kim et al.  To this end we merge the
sorted list $L'$ with the sorted list $L_{d,s_i}(S_i)$ and store the
result in $L'$.  We assume that the nodes in $L'$ are sorted in
increasing order with respect to their numerical sorting keys.  Ties
are broken in favor of nodes in $S_i$.  This can be achieved in linear
time $O(|S_i|+|S_j|)$.  Now recall that a node $u\in S_j$ imposes a
penalty $\pi(u)$ on $v\in S_i$ if $d_{T'}(u,v)>\rho(u)$ or
equivalently $d_{T'}(s_i,v)>\rho(u)-d_{T'}(u,s_i)$.  This is
tantamount to that $u$ precedes $v$ in $L'$.  Hence, in order to
compute $p'(v,S_j)$ for \emph{all} $v\in S_i$ it suffices to traverse
$L'$ once.  In doing so, one can maintain the sum of penalties of all
nodes $u\in S_j$ encountered so far, which equals the penalty
$p'(v,S_j)$ whenever a node $v\in S_i$ is reached.

The running time of this merge-and-scan operation is $O(|S_i|+|S_j|)$
since the necessary sorted lists have already been computed.  Thus we
can compute $p'(v,S)$ for all $v\in S$ in total time $O(|S|)$ once we
know the $p'$-values and respective lists for all child TTSBs of~$S$.

Note that the bottom of the recursion, that is, when $T'_{st}$
consists merely of the pair $(s,t)$ and $(t,s)$ of anti-parallel arcs
can trivially be handled constant time.

To sum up, this leads us to an algorithm whose running time $h(|S|)$
can be described by the following recurrence
\begin{displaymath}
  h(|S|)=O(|S|)+\sum_{i=1}^{k} h(|S_i|)\, ,
\end{displaymath}
where $k\leq 5$, $\sum_{i=1}^k |S_i|=|S|+4$ and $|S_i|\leq\frac12
|S|+1$.  This implies that $h(n)$ is $O(n\log n)$.

\begin{theorem}\label{thm:disc-single-cov-tree}
  The indirect covering subtree problem and hence also the single
  maximum coverage location can be solved in time $O(n\log n)$. \qed{}
\end{theorem}

\section{A Matching Lower Bound}\label{sec:matching-lower-bound}

In this section we complement our algorithm with a lower bound
$\Omega(n\log n)$ on the running time for solving single maximum
coverage location on a tree.  This shows that (for certain
computational models) our algorithm is optimal.

We make use of a recent result which is summarized in the following
theorem.
\begin{theorem}[\cite{ben-amram+galil:RAM}]\label{thm:galil}
  Let $W\subseteq\RR^n$.  If $W$ is recognized in time $t(n)$ on a
  real-number RAM that supports direct assignments, memory access,
  flow control, and arithmetic instructions $\{+,-,\times,/\}$ then
  $t(n)=\Omega(\log\beta(W^\circ))$.  \qed{}
\end{theorem}
Here, $W^\circ$ denotes the interior of $W$ and $\beta(W')$ denotes
the number of connected components of some set $W'\subseteq\RR^n$.

To prove our lower bound we introduce a variant of the set
disjointness problem.  To this end let $n\in \NN$.  The set
$W_n\subseteq \RR_+^{2n}$ contains all tuples
$(x_1,\ldots,x_n,y_1,\ldots,y_n)$ such that $x_1<\ldots<x_n$ and
$x_i\neq y_j$ for all pairs $i,j$.  Consider a permutation $\pi$ on
the set $\{1,\ldots,n\}$ and some tuple
$x_1<y_{\pi(1)}<x_2<y_{\pi(2)}<\ldots<x_n<y_{\pi(n)}$ in $W_n$.  It is
easy to see that for different permutations such tuples lie in
different connected components of $W_n^\circ$ so $W_n^\circ$ contains
at least $n!$ connected components.  Hence any RAM of the above
described type takes time $\Omega(n\log n)$ to recognize $W_n$.

We establish a linear time reduction from the problem to recognize
$W_n$ to the single maximum coverage location problem on a tree with
$O(n)$ nodes.  To this end consider a tuple
$(x_1,\ldots,x_n,y_1,\ldots,y_n)$ for which we want to decide whether
or not it is contained in $W_n$.

First we check if $x_1<\ldots<x_n$.  Then we create an edge $(u,v)$ of
some length $c(u,v)>\max\{\,x_i,y_i\mid i=1,\ldots,n\,\}$ and choose
some radius $\rho$ such that $\rho>c(u,v)$.  For any $y_i$ we create
two edges $(u,u_i)$ and $(v,v_i)$ of lengths $\rho-y_i$ and
$y_i+\rho-c(u,v)$, respectively.  Finally, we create for each $x_i$ a
node $\tilde x_i$ on edge $(u,v)$ with distance $d(u,\tilde
x_i):=x_i$.  For each node $z$ in the node set $V:=\{\,u_i,v_i,\tilde
x_i\mid i=1,\ldots,n\,\}\cup\{u,v\}$ we set $\pi(z):=1$ and
$\rho(z):=\rho$, which completes the reduction.

First suppose that we locate a facility outside the path $P(u,v)$.
Assume that the facility is located at some node $u_i$.  Then the
distance of $u_i$ to $u$ is positive and $d(u,v_j)\geq \rho$ for any
$j$.  Hence, none of the nodes $v_j$ is covered by $u_i$ and the
penalty cost imposed on $u_i$ must be at least $n$.  The case where is
the facility is placed at some node $v_i$ is treated analogously.

Now suppose for a moment that we can locate a facility anywhere at the
path $P(u,v)$, that is, also at interior points of edges on $P(u,v)$.
The point $x$ where the facility is located can then be identified
with the distance $d(u,x)$.  First, all nodes on $P(u,v)$ are covered
by $x$ since $\rho>d(u,v)$.  Due to our construction $x$ covers all
nodes $u_i$ where $x\leq y_i$ and all nodes $v_j$ where $x\geq y_j$.
Thus, the penalty imposed on $x$ is exactly $n$ if $x$ is not
contained in the set $\{y_1,\ldots,y_n\}$.  If $x=y_j$ then $x$ covers
both $u_j$ and $v_j$ and hence the penalty is bounded by $n-1$.  Since
the facility can only be placed at nodes $\tilde x_j$, that is, at
distances $x_i$ from $u$ we conclude that the minimum penalty cost is
$n$ if the input tuple $(x_1,\ldots,x_n,y_1,\ldots,y_n)$ lies in $W_n$
and $n-1$ otherwise.

\begin{theorem}\label{thm:matching-lower-bound}
  Any real-number RAM that complies with Theorem~\ref{thm:galil} takes
  at least $\Omega(n\log n)$ time to solve single maximum coverage
  location on a tree even for unit penalties and uniform radii. \qed{}
\end{theorem}

\section{Implications for Related Problems}\label{sec:impl-relat-probl}

The variant of single maximum coverage location where the facility can
be placed not only at the nodes but also at interior points of edges
is called the \emph{absolute} maximum coverage location problem.  Kim
et al.\ show \cite{kim-etal:location-tree-shaped} that a set of $O(n)$
critical points (that is a set of point which is guaranteed to to
contain an optimal point) for absolute single maximum coverage
location can be found in time $O(n\log n)$. We infer that also the
absolute variant can be solved in $O(n\log n)$ on a tree.

Another implication of our result leads us to the realm of
\emph{competitive location}.  Let a graph $G=(V,E)$ and $r,p\leq n$ be
given.  We assume that the graph is edge and node weighted.  Let
$X,Y\subseteq G$ be sets of nodes or interior points of edges.  Then
$w(Y\prec X)$ denotes the total weight $\sum\{w(u)\mid u\in V\text{
  and }d(u,Y)<d(u,X)\}$ of nodes that are closer to $Y$ than to $X$.
Given some point set $X$ the goal of the \emph{$(r,X)$-medianoid
  problem} \cite{Hakimi:competitive-environment} is to identify a set
$Y$ of $r$ points such that $w(Y\prec X)$ is maximized.  This maximum
weight is denoted by $w_r(X)$.  The goal of the \emph{$(r,p)$-centroid
  problem} is to find a $p$-element point set $X$ such that $w_r(X)$
is minimized.  

By setting $\rho(u):=d(v,X)-\epsilon$ (where $\epsilon$ is a suitably
small constant) and $\pi(u):=w(u)$ one can easily verify that
$(r,X)$-medianoid is a special case of the multiple maximum coverage
location problem with $r$ servers.  On general graphs the problem is
NP-hard \cite{Hakimi:competitive-environment}.  It can be solved
efficiently in $O(rn^2)$ on trees \cite{Tamir:p-median-tree}.  Our
result leads to an $O(n\log n)$ algorithm for the absolute and the
discrete version of $(1,X)$-medianoid on trees.
\begin{corollary}
  The discrete and the absolute $(1,X)$-medianoid problem can be
  solved in $O(n\log n)$ on trees. \qed{}
\end{corollary}
It is not hard to extend the lower bound provided by
Theorem~\ref{thm:matching-lower-bound} to $(1,X)$-medianoid.  Since
the radii of the tree constructed in the reduction are uniform and
hence all equal some number $\rho$, we can furnish each node $z$ on
that tree with a pendant leaf $z'$ at a distance
$d(z,z')=\rho+\epsilon$.  The set $X$ contains exactly those pendant
leaves.  It is clear that for any node $y$ in this enhanced tree $T'$
the gain $w(y\prec X)$ equals exactly $w(T)-p(y)$ in the original tree
$T$. This implies that both instances lead to the same optimum. Thus
also the algorithm for $(1,X)$-medianoid is optimal in terms of the
running time.

Now let's turn our view to the $(r,p)$-centroid problem.  The problem
is known to be $\Sigma^{\text{p}}_2$-complete on general graphs
\cite{spoerhase+wirth:multiple-voting} and NP-hard even on path graphs
\cite{spoerhase+wirth:rp-centroid}.  However, both the absolute and
the discrete variant of $(1,p)$-centroid on trees can be solved in
polynomial time $O(n^2\log^2 n\log w(T)/\log\log n)$ and $O(n^2\log^2
n\log w(T)\log D/\log\log n)$, respectively
\cite{spoerhase+wirth:rp-centroid}.  Here $D$ is the maximum edge
length of the input tree $T$.  Those algorithms rely on $O(n\log
w(T))$ (resp. $O(n\log w(T)\log D)$) calls to a subroutine solving
$(1,X)$-medianoid on a tree.

The algorithm provided here allows us to solve $(1,X)$-medianoid in
$O(n\log n)$ which yields.
\begin{corollary}
  The discrete and the absolute $(1,p)$-centroid problem for trees can
  be solved in $O(n^2\log n\log w(T))$ and $O(n^2\log n\log w(T)\log
  D)$, respectively. \qed{}
\end{corollary}

\section{Concluding Remarks}

We have provided an $O(n\log n)$ algorithm for solving the indirect
covering subtree problem which improves upon the previously best
algorithms for this problem and single maximum coverage location on
trees.  We have also shown that our algorithm is optimal for certain
unit-cost RAM models.  Our result leads also to an optimal algorithm
for $(1,X)$-medianoid and faster algorithms for $(1,p)$-centroid on
trees.

It would be interesting to identify larger problem classes of location
problems on trees where the the two-terminal subtree technique can be
applied.  It would also be worth investigating the existence of faster
algorithms on path graphs.

\end{document}